\documentclass{article}
\def\BibTeX{{\rm B\kern-.05em{\sc i\kern-.025em b}\kern-.08em
    T\kern-.1667em\lower.7ex\hbox{E}\kern-.125emX}}
\usepackage{graphics}
\usepackage{color}
\usepackage{url}
\usepackage{subcaption}
\usepackage{latexsym}
\usepackage{amsfonts,amssymb,amsmath}
\usepackage[dvips]{graphicx}
\usepackage{epsfig}
\usepackage{verbatim}   % provides comment and verbatim commands that work!
\usepackage{amsmath}
\usepackage{amssymb}
\usepackage{fancyhdr}
\usepackage[driver=pdftex]{geometry}
\usepackage{algorithmic}
\usepackage{algorithm}
\usepackage[utf8]{inputenc}
\usepackage{amsthm}
\allowdisplaybreaks
\usepackage{mathrsfs}
\usepackage[normalem]{ulem}
\usepackage{array}
\usepackage{tikz}
\usepackage{quantikz}
\usepackage[hidelinks]{hyperref}
\usepackage[nameinlink]{cleveref}
\usepackage{multicol}
\usepackage{cuted} 
\usepackage{booktabs}

\crefname{theorem}{thm.}{thm.}
\crefname{section}{sec.}{sec.}
\crefname{corollary}{cor.}{cor.}
\crefname{lemma}{lemma}{lemma}
\Crefname{theorem}{Thm.}{Thm.}
\Crefname{section}{Sec.}{Sec.}
\Crefname{figure}{Fig.}{Fig.}
\Crefname{equation}{Eqn.}{Eqn.}
\newtheorem{definition}{Definition}
\newtheorem{theorem}{Theorem}
\newtheorem{lemma}{Lemma}

\newtheorem{corollary}{Corollary}

\newcommand{\ev}{\mathbf{e}}
\newcommand{\bv}{\mathbf{b}}
\newcommand{\xv}{\mathbf{x}}

\newcommand{\psiv}{\mathbf{\psi}}
\newcommand{\phiv}{\mathbf{\phi}}

\newcommand{\Zev}{\mathbf{0}}

\newcommand{\Am}{\mathbf{A}}

\newcommand{\Zm}{\mathbf{Z}}
\newcommand{\Um}{\mathbf{U}}

\newcommand{\In}{\mathbf{I}}
\newcommand{\Id}{\mathbf{I}}

\newcommand{\Vm}{\mathbf{V}}

\newcommand{\bm}{\mathbf{b}}
\newcommand{\Dm}{\mathbf{D}}

\newcommand{\Sp}{\sigma_+}
\newcommand{\Sm}{\sigma_-}

\newcommand{\thetav}{\mathbf{\theta}}
\newcommand{\ra}{\rangle}
\newcommand{\la}{\langle}

\newcommand{\uv}{\mathbf{u}}

\usepackage{authblk}
\title{Efficient Variational Quantum Linear Solver for Structured Sparse Matrices}
\author[a,1]{Abeynaya Gnanasekaran}
\author[b]{Amit Surana}
\affil[a]{RTX Technology Research Center (RTRC), Berkeley, California, USA.}
\affil[a]{RTX Technology Research Center (RTRC), East Hartford, Connecticut, USA.}
\affil[1]{Corresponding author: abeynaya.gnanasekaran@rtx.com}
\begin{document}
\maketitle

\begin{abstract}
We develop a novel approach for efficiently applying variational quantum linear solver (VQLS) in context of structured sparse matrices. Such matrices frequently arise during numerical solution  of partial differential equations which are ubiquitous in science and engineering. Conventionally, Pauli basis is used for linear combination of unitary (LCU) decomposition of the underlying matrix to facilitate the evaluation the global/local VQLS cost functions. However, Pauli basis in worst case can result in number of LCU terms that scale quadratically with respect to the matrix size. We show that by using an alternate basis one can better exploit the sparsity and underlying structure of matrix leading to number of tensor product terms which scale only logarithmically with respect to the matrix size. Given this new basis is comprised of non-unitary operators, we employ the concept of unitary completion to design efficient quantum circuits for computing the global/local VQLS cost functions. We compare our approach with other related concepts in the literature including unitary dilation and measurement in Bell basis, and discuss its pros/cons while using VQLS applied to Heat equation as an example.  
\end{abstract}

\section{Introduction}
Variational quantum algorithms (VQAs) are hybrid classical-quantum algorithms that have emerged as promising candidates to optimally utilize today's Noisy Intermediate Scale Quantum (NISQ) devices. VQAs use a quantum subroutine to evaluate a cost/ energy function using a shallow parameterized quantum circuit. A classical optimizer is used to update the parameters using the output measurement statistics. Typically, the quantum subroutine performs the most expensive part of the overall algorithm to achieve quantum advantage over a fully classical algorithm. VQAs provide a versatile framework to tackle a wide variety of problems. As such, they have been successfully used in a wide variety of applications such as finding ground and excited states of Hamiltonians~\cite{app1,app2,app3,app4}, dynamical quantum simulation~\cite{ds1,ds2}, combinatorial optimization~\cite{c1,c2}, solving linear systems~\cite{VQLS,ls2,ls3}, integer factorization~\cite{f1}, principal component analysis~\cite{pca1,pca2,pca3}  and quantum machine learning \cite{biamonte2017quantum}.

Despite tremendous growth and development in the field of VQAs, many challenges remain for obtaining quantum speedups when scaling the system size on NISQ devices. The three main challenges include trainability of the quantum circuit, efficiency of cost function evaluations and its accuracy including the effect of noise. (1) Trainability: The presence of barren plateau in the cost function landscape is one of the main bottlenecks in VQAs~\cite{bp1}. Exponential precision is required to escape barren plateau requiring exponential computational scaling. The use of structured ansatz~\cite{bp_an1,bp_an2}, local cost functions~\cite{bp2,pca3}, optimal parameter initialization~\cite{bp_param}, and tailored optimizers~\cite{bp_opt} have been shown to address this challenge. (2) Efficiency: Cost function evaluations are key to VQAs and efficiently estimating them is essential to demonstrate quantum benefits. Since, the cost function Hamiltonians are not necessarily unitary operators, a common technique is to decompose them into a linear combination unitaries (LCU) formed from tensor product of Pauli operators. However, this decomposition can contain large number of terms, requiring several circuit measurements to evaluate the cost function, thereby reducing the  overall efficiency. The use of commuting operator sets~\cite{comm}, optimized measurement sampling~\cite{samp}, classical shadows~\cite{shad} and neural network tomography~\cite{tomo} have been proposed to improve efficiency. (3) Accuracy: While VQAs are designed to use for NISQ devices, hardware noise can still slow the training process and modify the global optimum. Inherent noise resilience of VQAs is an active area of research \cite{comm,co_exp,zne,zne_pnc,trex}.  

In this work, we focus on improving the efficiency of VQAs by reducing the number of circuit evaluations/measurements required for evaluating the cost function. A common approach to deal with this issue is to partition the LCU terms into subsets of commuting operators that are simultaneously measurable~\cite{comm}.In contrast, the authors in~\cite{liu2021variational} developed an efficient VQA for solving the Poisson equation using the Variational Quantum Linear Solver (VQLS) framework. They develop a novel decomposition of the Possion coefficient matrix under a set of simple (non-unitary) operators and construct observables to evaluate the expectation values of those operators on a quantum computer. Through this approach, they demonstrate an exponential reduction in the number of terms needed in tensor product decomposition as compared to the Pauli basis for the Poisson coefficient matrix. They also extend their technique to tridiagonal and pentadiagonal Toeplitz matrices that commonly arise in solving partial differential equations (PDEs). We highlight some drawbacks of this technique:
\begin{itemize}
\item Local cost functions are necessary to reduce the effect of barren plateaus and improve trainability of the circuit. However, there is no straightforward extension of their technique to compute terms in the \textit{local} VQLS cost function.
\item The quantum circuit to compute the expectation values requires the measurement of every qubit increasing the measurement overhead.
\end{itemize}
In this work, we address these drawbacks by developing a novel quantum circuit to evaluate the expectation values of these non-unitary operators for both the global and local VQLS cost functions. 

The rest of the paper is organized as follows. \Cref{sec: vqls} gives an overview of the VQLS framework, cost Hamiltonians and the LCU technique. In \Cref{sec: sigma}, we discuss efficient decomposition under the set of simple non-unitary operators by considering a model linear system arising from discretization of the heat equation. In \Cref{sec: method}, we develop our novel approach for evaluating the cost functions and provide details on the circuit construction and resource estimation. In \Cref{sec: discussion} we compare our approach to other related techniques from the literature, and finally conclude in the \Cref{sec: conc}.

\paragraph*{Notation} We will denote by $\mathbb{R}$ as the set of real numbers, $\mathbb{C}$ as the set of complex numbers,  small bold letters as vectors, capital bold letters as matrices/operators,  $\Am^*$ as the vector/matrix complex conjugate, $\Am^T$ as the vector/matrix transpose and $\Id_s$ as an Identity matrix of size $s \times s$. We will use standard braket notation in representing the quantum states. We use $C^n X$ to represent the \textit{n}-controlled Toffoli gate. With this notation, $C^1 X$ is the CNOT gate and $C^2 X$ is the Toffoli gate or CCNOT gate. 

\section{Variational Quantum Linear Solver}
\label{sec: vqls}

VQLS is a variational algorithm that solves linear system of the form $\Am \xv = \bv$ where $\Am$ is a $N\times N$ complex valued matrix and $\bv\in \mathbb{C}^N$ \cite{VQLS}. Without loss of generality we will assume $N=2^n$ for some integer $n$. Given an efficient unitary operator $\Um$ that prepares the quantum state $\ket{\bm}$ and a LCU decomposition of $\Am$ in terms of $n_l$ unitary operators $\Am_l$, i.e., 
\begin{equation}\label{eq:LCU}
\Am = \sum_{l=1}^{n_l} \alpha_l \Am_l,
\end{equation}
the VQLS algorithm outputs a state $\ket{\psi}$ that is approximately proportional to the solution of the linear system. To accomplish this, the algorithm encodes the solution through an ansatz $\Vm(\thetav)$, such that $\ket{\psi} = \Vm(\thetav) \ket{0}$. The parameters $\thetav$ are optimized in a hybrid classical-quantum loop until the value of a pre-defined cost function $C(\thetav)$ is below a user-defined threshold. 

One can define two kinds of cost functions: global and local \cite{VQLS}. The global cost function is defined as,
\begin{align*}
C^{ug}(\thetav)&=Tr(|\phiv\ra\la\phiv|(\In-|\bv\ra\la\bv|))=\la \psiv|H^{g}|\psiv\ra,
\end{align*}
where, $|\phiv\ra=\Am|\psiv\ra $ and
$
H^{g}=\Am^*\Um(\In-|\Zev\ra\la\Zev|)\Um^*\Am,
$
is the effective global Hamiltonian. One can also define a normalized version of the cost function as,
\begin{align*}
C^{g}&=\frac{C^{gu}}{\la\phiv|\phiv\ra}=1-\frac{|\la\bv|\phiv\ra|^2}{\la\phiv|\phiv\ra}.
\end{align*}
$C^{g}$ can be calculated by computing $\la\bv|\phiv\ra$ and $\la\phiv|\phiv\ra$ using the LCU of matrix $\Am$ as,
\begin{align}
  \la\phiv|\phiv\ra &= \sum_{i,j=1}^{n_l}\beta_{ij}\alpha_i\alpha_j^*, \quad
  |\la\bv|\phiv\ra|^2  =\sum_{i,j=1}^{n_l}\gamma_{ij}\alpha_i\alpha_j^*, \nonumber\\
\beta_{ij} &=\la \Zev|\Vm^*\Am_j^*\Am_i\Vm|\Zev\ra, \label{eq:beta}\\ 
\gamma_{ij} &=\la \Zev|\Um^*\Am_i\Vm|\Zev\ra \la \Zev|\Vm^*\Am_j^*\Um|\Zev\ra.\label{eq:gamma}
\end{align}

Global cost functions defined above can exhibit barren plateaus which drastically affect the trainability of the cost function \cite{VQLS}. To alleviate this issue, local cost functions has been proposed as follows
\begin{equation*}
C^{ul}=\la\psi |H^l| \psi\ra,
\end{equation*}
where, the local Hamiltonian $H^l$ is defined as
\begin{equation*}
H^l=\Am^*\Um\left(\In-\frac{1}{n}\sum_{k=1}^{n}|0_k\ra\la 0_k|\otimes \In_{\tilde{k}}\right)\Um^*\Am, 
\end{equation*}
where, $|0_k\ra\la 0_k|\otimes \In_{\tilde{k}}=\In_2\otimes \cdots |0_k\ra\la 0_k|\cdots  \otimes \In_2$ with $\Id_{n}$ being an identity matrix of size $n \times n$.
One can also define local version of normalized cost,  
\begin{align*}
&C^{l}=\frac{C^{ul}}{\la\phiv|\phiv\ra}=1-\frac{1}{n}\frac{1}{\la\phiv|\phiv\ra}\sum_{k=1}^{n}\Sigma_k.
\end{align*}
where the term,
\begin{align*}
\Sigma_k &=\la\mathbf{0}|\Vm^*\Am^*\Um(|0_k\ra\la 0_k|\otimes \In_{\tilde{k}})\Um^*\Am\Vm|\mathbf{0}\ra, \\
&= \sum_{i,j=1}^{n_l} \frac{\alpha_i\alpha_j^*}{2} \left(\beta_{ij}+\delta_{ijk}\right),\nonumber
\end{align*}
with,
\begin{align}
\delta_{ijk}&=\la \Zev|\Vm^*\Am_j^*\Um(\Zm_k\otimes \In_{\tilde{k}})\Um^*\Am_i\Vm|\Zev\ra\label{eq:delta}.
\end{align}
Typically, the LCU decomposition (\ref{eq:LCU}) is done in terms of Pauli operators $P = \{I, \sigma_x, \sigma_y, \sigma_z\}$, so that $\Am_l=\sigma_1\otimes\sigma_2\otimes\cdots\otimes\sigma_n$ with $\sigma_i\in P, i=1,\cdots,n$. $\Am_l$ is a unitary by definition and can be implemented efficiently on a quantum device.
However, for structured sparse matrices that typically arise in solving PDEs, number of terms $n_l$ in the LCU decomposition can be much larger than the number of non-zero entries in the matrix. Since, $n_l^2 \log N$ terms are needed to calculate the local cost function, it is important to reduce the number of terms. Ideally, one would like $n_l = \mathcal{O}(\text{poly}(\log N))$.

\section{Decomposition Using Sigma Basis}
\label{sec: sigma}

In~\cite{liu2021variational}, the authors develop an alternate tensor product decomposition under a set of simple, albeit non-unitary, operators which we refer to as the sigma basis.
\begin{definition}
The sigma basis is a set
\begin{equation}\label{def:S}
S =\{\Id, \Sp, \Sm, \Sp\Sm, \Sm\Sp\},
\end{equation}
where, $\Sp=\ket{0} \bra{1}$, $\Sm=\ket{1}\bra{0}$, $\Sp\Sm = \ket{0}\bra{0}$ and $\Sm\Sp=\ket{1}\bra{1}$.
\end{definition}
Using this basis, they developed an efficient tensor product decomposition of the $N \times N$ coefficient matrix arising from a finite difference discretization of the Poisson equation into $(2\log N +1)$ terms. They compute the VQLS global cost function by designing observables to estimate $\beta_{ij}$ and $\gamma_{ij}$ terms and performing measurements in the Bell basis. Through this work, they demonstrated that the tensor product decomposition can be done even in terms of non-unitary operators as long as a quantum circuit can be constructed efficiently to evaluate the cost function.

We adapt the sigma basis in our work and develop a novel quantum circuit that can be extended to compute VQLS terms in both the global and local cost functions. Our technique can be understood through the lens of unitary completion and does not involve designing observables, see Section \ref{sec: discussion} for comparison with the approach discussed above. We begin by illustrating our approach for Heat equation since it provides another PDE example (in addition to Possion equation) where sigma basis leads to an efficient decomposition. However, note that our circuit construction is general and can be applied to any matrix given as a tensor product decomposition in terms of sigma basis.

\subsection{Model System: Heat Equation}
\label{subsec: heat}
Consider the 1D heat transfer problem over a domain $[0,\, l]$ with Neumann boundary conditions as follows:
\begin{align}
\frac{\partial u}{\partial t}&=\alpha \frac{\partial^2 u}{\partial x^2},\nonumber\\
  -k\frac{\partial u}{\partial x}\bigg|_{x=0} = q, &\,
  -k\frac{\partial u}{\partial x}\bigg|_{x=l} = 0, \nonumber\\
   u(x,0)&=u_0(x),
\end{align}
where, $\alpha$ is the thermal diffusivity, $k$ is the thermal conductivity of the material and $q$ is a constant heat flux. We discretize the PDE in space with second order accuracy, and the boundary condition with first order accuracy. Backward Euler scheme is used for time discretization. This leads to system of linear equations of the form
\begin{align}
\label{eq: diff}
\bigg(\Id_{n_x}-\frac{\alpha \Delta t}{(\Delta x)^2} \Am'\bigg)\uv_{t+1}= \uv_t+ \frac{q\,\Delta t}{k\Delta x}\ev_1,
\end{align}
where, $n_x$ is the number of spatial grid point, $\uv_t = \begin{bmatrix}
u_{1,t} & u_{2,t} & \dots & u_{n_x ,t}
\end{bmatrix}^T$, $u_{i,t} = u((i-1)\Delta x, t\Delta t)$, $\Delta x$ is the spatial grid size, $\Delta t$ is the temporal grid size,  $\ev_1=\begin{bmatrix} 1 & 0 & 0 \cdots 0 \end{bmatrix}^T\in \mathbb{R}^{n_x}$, and $\Am'$ is a $n_x \times n_x$ matrix of the form
\begin{align*}
\Am' &= \begin{bmatrix*}[r]
 -1 & 1 &  &  & 0 \\
       1 & -2 & 1 &  &  \\
       & \ddots & \ddots & \ddots &  \\
       &  & \ddots  & -2 & 1 \\
      0 &  &  & 1 & -1 \\
\end{bmatrix*}.
\end{align*}
We can express the difference equations of the form~\Cref{eq: diff} for $t=1,\cdot,n_t$ into a single linear system $\Am \uv = \bv$, where, $\uv = \begin{bmatrix}
\uv_1^T & \uv_2^T & \dots & \uv^T_{n_t}
\end{bmatrix}^T$, $\bv = \begin{bmatrix}
\uv_0^T & \frac{q\Delta t}{k \Delta x}\ev^T_1 & \dots & \frac{q\Delta t}{k \Delta x}\ev^T_1
\end{bmatrix}^T$ and 
\begin{align}\label{eq: full_coeff}
\Am &= \begin{bmatrix*}[r]
\Id_{n_x} & & & 0\\
-\Id_{n_x} & \Id_{n_x} & & \\
& \ddots & \ddots & \\
0 & & -\Id_{n_x} & \Id_{n_x}
\end{bmatrix*} - \frac{\alpha \Delta t}{\Delta x^2} \begin{bmatrix}
0 & & & \\
 & \Am' & & \\
 &&\ddots & \\
 && &\Am'
\end{bmatrix} = \Am_1 - \frac{\alpha \Delta t}{\Delta x^2}  \Am_2.
\end{align}

We can find a tensor decomposition of $\Am_1$ and $\Am_2$ under the sigma basis using a recursive decomposition strategy as follows. For simplicity assume $n_x= 2^s$ and $n_t = 2^t$. Then $\Am_1$ corresponds to $s+t$ qubits and can be written as
 \[\Am_1 \coloneqq \Am_1^{(s+t)} = \begin{bmatrix}
 \Am_1^{(s+t-1)} & 0 \\
 \Dm_1^{(s+t-1)} & \Am_1^{(s+t-1)}
 \end{bmatrix},
\]
where
\begin{align*}
\Dm_1^{(s+t-1)} &= \begin{bmatrix}
0 & \dots & -\Id_{n_x} \\
\vdots & \ddots & \vdots \\
0 & \dots & 0
\end{bmatrix} = \underbrace{\Sp \otimes \dots \otimes \Sp}_{t-1 \text{ times}} \otimes \Id_{n_x}.
\end{align*} 
Then, 
\begin{align*}
\Am_1^{(s+t)} &= \Id_2 \otimes \Am_1^{(s+t-1)} + \Sm \otimes \Dm_1^{(s+t-1)}, \\
\Am_1^{(s+1)} &= \begin{bmatrix*}[r]
\Id_{n_x} & 0\\
-\Id_{n_x} & \Id_{n_x}
\end{bmatrix*} = \Id_2 \otimes \Id_{n_x} - \Sm \otimes \Id_{n_x}. 
\end{align*}
With this recursive relation, we $\Am_1$ can be written using $\log n_t + 1$ terms. Next, $\Am_2$ can be written as
\begin{align*}
\Am_2 &= \Id_{n_t} \otimes \Am' - \Sp\Sm \otimes \dots \otimes \Sp\Sm \otimes \Am',
\end{align*}
with
\begin{align*}
\Am' &= \begin{bmatrix*}[r]
 -2 & 1 &  &   0 \\
       1 & \ddots &\ddots &    \\
       & \ddots & \ddots & 1   \\
      0 &    & 1 & -2 \\
\end{bmatrix*}  + \begin{bmatrix*}[r]
 1 &  &  &    \\
        &0 & &    \\
       &  & \ddots &    \\
       &   &  & 1 \\
\end{bmatrix*} \\
&= \Am'_1 + \Am'_2. 
\end{align*}
$\Am_2'$ is simply two terms $\Sp\Sm \otimes \dots \otimes \Sp\Sm + \Sm\Sp \otimes \dots \otimes \Sm\Sp$. We can get a recursive decomposition for $\Am_1'$ using similar procedure as for $\Am_1$ above. Specifically, it follows that
\begin{align*}
\Am_1' \coloneqq \Am_1'^{(s)} &= \Id_2 \otimes \Am_1'^{(s-1)} + \Sm \otimes \Dm'^{(s-1)} + \Sp \otimes (\Dm'^{(s-1)})^T, \\
\Dm'^{(s-1)} &= \underbrace{\Sp \otimes \dots \otimes \Sp}_{s-1 \text{ times}}, \\
\Am_1'^{(1)} &= -2\,\Id_2 + \Sm + \Sp.
\end{align*}
Thus, we can write $\Am'$ using $2\, \log n_x +3$ terms and hence, $\Am_2$ using $4 \,\log n_x + 6$ terms.

This non-trivial example demonstrates that tensor decomposition under the sigma basis can lead to logarithmic or poly-logarithmic number of terms for structured sparse matrices. We obtain a similar decomposition given a Robin boundary condition of the form $w_1 u(x,t) + w_2 \frac{\partial u}{\partial x} = q$ with the only difference being the non-zero entries of $\Am_2'$ taking the value $w_2 / (w_1 \Delta x + w_2)$.

\section{Cost Function Evaluation}
\label{sec: method}
In this section, we develop novel quantum circuits to evaluate the terms $\beta_{ij}$, $\gamma_{ij}$, and $\delta_{ijk}$ that appear in the global and local VQLS cost functions, for $\Am \in \mathbb{R}^n\otimes \mathbb{R}^n$ given in form of tensor product decomposition defined over sigma basis, i.e., 
\begin{equation}
\Am=\sum_{i=1}^{n_l}\alpha_i\Am_i,
\end{equation} 
where, $\alpha_i\in \mathbb{C}$,  $\Am_i=\sigma_1\otimes\cdots\otimes\sigma_n$ with $\sigma_i\in S$ for $i=1,\cdots,n$.

\subsection{Quantum Circuits for $\delta_{ijk}$ and $\beta_{ij}$}
First, let us focus on the terms $\delta_{ijk}$ as defined in~\Cref{eq:delta}. Consider a unitary operator $\Um_l$ associated with operator $\Am_l$ of the form
\begin{equation}
\label{eq: Um_l}
\Um_l = \begin{bmatrix}
\Am^c_l & \Am_l \\
\Am_l & \Am^c_l
\end{bmatrix},
\end{equation}
where, $\Am_l^c$ is the unitary complement of operator $\Am_l$. This essentially means that the subspace spanned by the columns (rows) or $\Am^c_l$ are orthogonal to the subspace spanned by the columns (rows) of $\Am_l$, see Section \ref{subsec:Um} for details. With this definition, it is straightforward to see that $\Um_l$ is in fact a unitary matrix. The application of $\Um_l$ on an ancilla system of the form $\ket{0}\ket{\psi}$ gives, 
\begin{equation}
\label{eq: Umact}
\Um_l \ket{0} \ket{\psi} = \ket{0} \Am^c_l \ket{\psi} + \ket{1} \Am_l\ket{\psi}.
\end{equation}
This is related to the concepts of unitary dilation, block encoding, projective measurements and unitary dynamics~\cite{nielsen2010quantum}, for further discussion see the Section~\ref{sec: discussion}. Starting with a $n+2$ qubit system (2 ancilla bits), we can compute $\delta_{ijk}$ through the following operations
\begin{align*}
\ket{0^{n+2}} \xrightarrow[]{H_{a_0}, \Vm(\theta)} &\frac{1}{\sqrt{2}} \big( \ket{00} \ket{\psi} +  \ket{10} \ket{\psi} \big)  \\
\xrightarrow[]{C_{\Um_i}, OC_{\Um_j}} &\frac{1}{\sqrt{2}} \big( \ket{0} {\Um}_j \ket{0}\ket{\psi} +  \ket{1} {\Um}_i \ket{0}\ket{\psi} \big)  \\
\qquad \quad =& \frac{1}{\sqrt{2}} \big( \ket{01} \Am_j \ket{\psi} + \ket{00} \Am_j^c \ket{\psi} + \ket{11} \Am_i \ket{\psi} +  \ket{10} \Am_i^c \ket{\psi} \big)  \\
\xrightarrow[]{\Um CC_{(\Zm_k\otimes \In_{\tilde{k}})} \Um^*} &\frac{1}{\sqrt{2}} \big( \ket{01} \Am_j \ket{\psi} +  \ket{00} \Am_j^c \ket{\psi} + \ket{11} \Um (\Zm_k\otimes \In_{\tilde{k}})\Um^*\Am_i \ket{\psi} +  \ket{10} \Am_i^c \ket{\psi} \big)  \\
\xrightarrow[]{H_{a_0}} &\frac{1}{2} \bigg( \ket{01} \big(\Am_j \ket{\psi} + \Um (\Zm_k\otimes \In_{\tilde{k}})\Um^*\Am_i \ket{\psi} \big) +   \ket{00} \big( \Am_j^c \ket{\psi} +  \Am_i^c \ket{\psi}  \big) + \\
  & \qquad \ket{11} \big( \Am_j \ket{\psi} - \Um (\Zm_k\otimes \In_{\tilde{k}})\Um^*\Am_i \ket{\psi} \big) +  \ket{10} \big(\Am_j^c \ket{\psi} -  \Am_i^c \ket{\psi}\big) \bigg),
\end{align*}
and, finally measuring the two ancilla bits 
\begin{align*}
P_{01} - P_{11} &= \frac{1}{4} \bigg( \Big\| \Am_j \ket{\psi} + \Um (\Zm_k\otimes \In_{\tilde{k}})\Um^*\Am_i \ket{\psi}  \Big\| ^2 -  \Big\|  \Am_j \ket{\psi} - \Um (\Zm_k\otimes \In_{\tilde{k}})\Um^*\Am_i \ket{\psi}  \Big\| ^2\bigg) \\
&= \bra{\psi} \Am_j^* \Um(\Zm_1\otimes \In_{\tilde{1}})\Um^*\Am_i \ket{\psi}.
\end{align*}
Here $C_*, OC_*, CC_*$ represent controlled, open-controlled, and control-control application of the unitaries. The associated quantum circuit is shown in \Cref{fig: local_cost_circuit}. The unitary operator $\Um_l$ can be implemented efficiently as shown in the blue boxes in \Cref{fig: local_cost_circuit}, details are given in the Section~\ref{subsec:Um}.  In order to evaluate $\beta_{ij}$, the circuit in \Cref{fig: local_cost_circuit} can be simplified by removing the blocks corresponding to $\Um (\Zm_k\otimes \In_{\tilde{k}})\Um^*$.

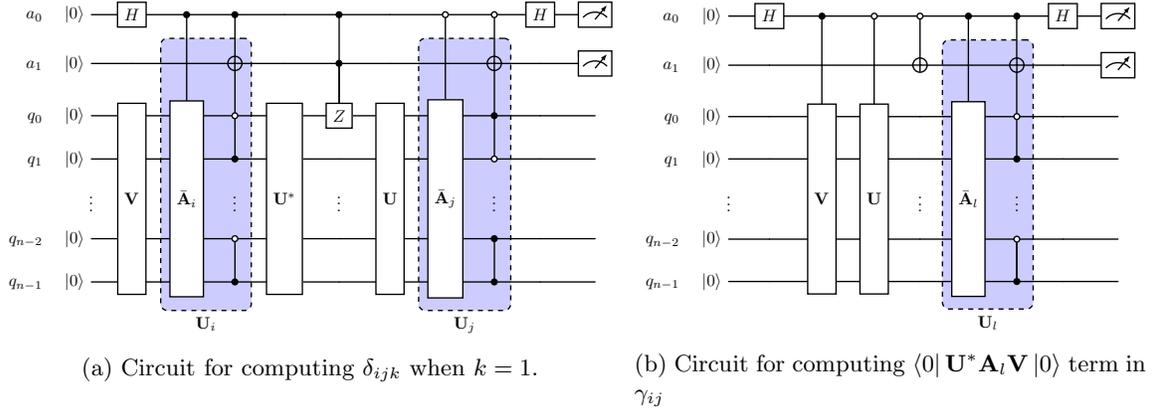
\begin{figure*}[h]
\centering
\begin{subfigure}[t]{0.55\textwidth}
\resizebox{1.\textwidth}{!}{
\begin{quantikz}[wire types = {q, q, q, q, n, q, q}, transparent]
\lstick{$a_0$ \quad \ket{0}}& \gate{H} & \ctrl{2} & \ctrl{1}  & & \ctrl{2} & &  \octrl{2} & \octrl{1} & \gate{H} & \meter{}  \\
\lstick{$a_1$ \quad \ket{0}}&           &    \gategroup[6,steps=2,style={dashed,rounded
corners,fill=blue!20, inner
xsep=2pt},background, label style={label
position=below,anchor=north,yshift=-0.2cm}]{{\sc
${\Um}_i$}}    & \targ{} & & \ctrl{1} & & \gategroup[6,steps=2,style={dashed,rounded
corners,fill=blue!20, inner
xsep=2pt},background, label style={label
position=below,anchor=north,yshift=-0.2cm}]{{\sc
${\Um}_j$}} & \targ{} & & \meter{}  \\
\lstick{$q_0$ \quad \ket{0}}& \gate[5]{\Vm}&  \gate[5]{\bar{\Am}_i}&  \octrl{-1}  & \gate[5]{\Um^*} & \gate{Z} & \gate[5]{\Um} & \gate[5]{\bar{\Am}_j} & \ctrl{-1}& & \\
\lstick{$q_1$ \quad \ket{0}} & & & \ctrl{-1}   &&& & &\octrl{-1}& &\\
\vdots &&& \vdots  && \vdots & && \vdots &  & \\
\lstick{$q_{n-2}$ \quad \ket{0}}&& & \octrl{0} && & & &\ctrl{0} & &\\
\lstick{$q_{n-1}$ \quad \ket{0}}&& & \ctrl{-1} && & & &\ctrl{-1} & &
\end{quantikz}
}
\caption{Circuit for computing $\delta_{ijk}$ when $k=1$. }
\label{fig: local_cost_circuit}
\end{subfigure}%
~
\begin{subfigure}[t]{0.45\textwidth}
\resizebox{1.\textwidth}{!}{
\begin{quantikz}[wire types = {q, q, q, q, n, q, q}, transparent]
\lstick{$a_0$ \quad \ket{0}}& \gate{H} & \ctrl{2} & \octrl{2} &  \octrl{1} &\ctrl{2} & \ctrl{1} &  \gate{H} & \meter{}  \\
\lstick{$a_1$ \quad \ket{0}}&   &&       &   \targ{}      & \gategroup[6,steps=2,style={dashed,rounded
corners,fill=blue!20, inner
xsep=2pt},background,  label style={label
position=below,anchor=north,yshift=-0.2cm}]{{\sc ${\Um}_l$}}   & \targ{} && \meter{}  \\
\lstick{$q_0$ \quad \ket{0}}& & \gate[5]{\Vm} & \gate[5]{\Um}& &  \gate[5]{\bar{\Am}_l} &  \octrl{-1}  &  & \\
\lstick{$q_1$ \quad \ket{0}} && & && &\ctrl{-1}  & & \\
\vdots & &&&\vdots && \vdots&   &\\
\lstick{$q_{n-2}$ \quad \ket{0}}& & && & &\octrl{1}  && \\
\lstick{$q_{n-1}$ \quad \ket{0}}& &&& &  &\ctrl{-1} & & 
\end{quantikz}
}
\caption{Circuit for computing $\bra{0}\Um^* \Am_l \Vm\ket{0}$ term in $\gamma_{ij}$}
\label{fig: global_cost_circuit}
\end{subfigure}
\caption{Hadamard test for computing the terms in the global and local cost functions. The circuit in (a) can be simplified by removing the blocks corresponding to $\Um (\Zm_k\otimes \In_{\tilde{k}})\Um^*$ to compute $\beta_{ij}$.}
\end{figure*}

\subsection{Quantum Circuit for $\gamma_{ij}$ }
 
The terms  $\gamma_{ij}$ defined in (\Cref{eq:gamma}) can be determined by computing $n_l$ terms of the form $\la \Zev|\Um^*\Am_l\Vm|\Zev\ra$ using the circuit shown in the~\Cref{fig: global_cost_circuit}.  Starting with a $n+2$ qubit system, the circuit performs the following sequence of operations:
\begin{align*}
\ket{0^{n+2}} \xrightarrow[]{H_{a_0}, OC_\Um, C_{\Vm(\theta)}}& \frac{1}{\sqrt{2}} \big( \ket{00} \ket{\bv} +  \ket{10} \ket{\psi} \big)  \\
\xrightarrow[]{\Um_l}& \frac{1}{\sqrt{2}} \big( \ket{00} \ket{\bv} +  \ket{10} \Um_l \ket{\psi} \big)  \\
=&\frac{1}{\sqrt{2}} \big( \ket{00} \ket{\bv} + 
  \ket{11} \Am_l \ket{\psi} +  \ket{10} \Am_l^c \ket{\psi} \big)  \\
 \xrightarrow[]{OC_X} &\frac{1}{\sqrt{2}} \big( \ket{01} \ket{\bv} + 
  \ket{11} \Am_l \ket{\psi} +  \ket{10} \Am_l^c \ket{\psi} \big)  \\
\xrightarrow[]{H_{a_0}} &\frac{1}{2} \bigg( \ket{01} \big(\ket{\bv} + \Am_l \ket{\psi} \big) +   \ket{00}  \Am_l^c \ket{\psi} +  \ket{11} \big( \ket{\bv} -\Am_l \ket{\psi} \big) -  \ket{10}  \Am_l^c \ket{\psi} \bigg).
\end{align*}
After measuring the two ancilla bits and using the relation
\begin{align*}
P_{01} - P_{11} &= \frac{1}{4} \bigg( \Big\| \ket{\bv} + \Am_l \ket{\psi}  \Big\| ^2 -  \Big\|  \ket{\bv} - \Am_l \ket{\psi}  \Big\| ^2\bigg) \\
&= \bra{\bv} \Am_l \ket{\psi} = \bra{0}\Um^* \Am_l \Vm\ket{0},
\end{align*}
leads to the desired result.

One can also construct a Hadamard Overlap test circuit in a similar fashion to avoid the controlled application of the unitaries $\Um, \Vm(\theta)$ to estimate $\gamma_{ij}$, see \cite{VQLS} for the details.  

\subsection{Construction of the quantum circuit for unitary operator $\Um_l$}
\label{subsec:Um}

To develop the quantum circuit for $\Um_l$, we need a few definitions and properties of the sigma basis $S$. The unitary completion and complement of a linear operator that preserves inner product in a subspace are defined in \Cref{def: completion}. Using this, we can prove some important properties of sigma basis in \Cref{thm: prop_sig} and by extension for $\Am_l$ in \Cref{cor: prop_A}. 
\begin{definition}
\label{def: completion}
 If $U: W \rightarrow V$ is a linear operator which preserves inner products, i.e., for any $\ket{w_1}$, $\ket{w_2}$ in $W \subset V$, $\bra{w_1}U^*U\ket{w_2} = \bra{w_1}w_2\ra$, then an unitary operator $\bar{U}: V\rightarrow V$ is its unitary completion if $\bar{U}$ spans the whole space $V$ and $\bar{U} \ket{w} = U \ket{w}\, \forall \ket{w} \in W$. Also, $U^c \coloneqq U'-U$ is the unitary complement of $U$. Such a unitary operator $\bar{U}$ always exists (see Ex 2.67,~\cite{nielsen2010quantum}).
\end{definition}

\begin{theorem}
\label{thm: prop_sig}
Properties of sigma basis:
\begin{enumerate}
\item The unique unitary completion for $\sigma_+$ and $\sigma_-$ is $\sigma_x = \begin{bmatrix}
0 & 1 \\
1 & 0
\end{bmatrix}$ (Pauli X) and for $\mathbf{I}, \sigma_+\sigma_-, \sigma_-\sigma_+$ is the Identity matrix $\Id$. 
\item The unitary complement of $\Sp$ is $\Sm$, $\Sm$ is $\Sp$, $\Sp\Sm$ is $\Sm\Sp$, $\Sm\Sp$ is $\Sp\Sm$ and for $\Id$ is $\mathbf{0}$.
\item $\sigma_k \sigma_k^T$ and $\sigma_k^T \sigma_k$ for $\sigma_k \in  S$  are idempotent matrices. 
\item $\sigma_k \sigma_k^T \sigma_k = \sigma_k$ for $\sigma_k \in  S$.  
\end{enumerate}
\begin{proof}
\hspace{1pt}
\begin{enumerate}
\item It is trivial to see $\sigma_x, \Id$ span $\mathbb{R}^2$. Consider subspace $W_+ = \{\alpha \ket{1}: \alpha \in \mathbb{R}\} \subset \mathbb{R}^2 $ such that for $w \in W_+$, $\bra{w} \Sp^T \Sp \ket{w} = \la w \ket{w}$. Then $\sigma_x \ket{w} = \Sp \ket{w} \quad \forall w \in W_+$ and hence by definition, $\sigma_x$ is an unitary completion of $\Sp$. By considering the same subspace $W_+$, we can show that $\Id$ is an unitary completion of $\Sm\Sp$. Similarly, by considering subspace $W_- = \{\alpha \ket{0}: \alpha \in \mathbb{R} \}$, one can show that $\sigma_x$ is also an unitary completion of $\Sm$, and $\Id$ is an unitary completion of $\Sp\Sm$. 
\item Follows by definition. 
\item It is straightforward to see that $(\sigma_k \sigma_k^T)^2 = \sigma_k \sigma_k^T$ as $\sigma_k \sigma_k^T \in \{\Id,\sigma_{+}\sigma_{-},\sigma_{-}\sigma_{+}\}$. Similarly for $\sigma_k^T \sigma_k$.
\item Trivially true for $\sigma_k \in \{\Sp\Sm, \Sm\Sp, \Id \}$. For $\sigma_k \in \{\Sp, \Sm\} $,
\begin{align*}
\Sp \Sp^T \Sp &= (\Sp\Sm)\Sp = \Sp, \\
\Sm \Sm^T \Sm &= (\Sm\Sp)\Sm = \Sm.
\end{align*}
\end{enumerate}
\end{proof}
\end{theorem}

\begin{corollary}
\label{cor: prop_A}
If $\Am_l = \otimes_{k} \sigma_k$ where $\sigma_k \in S$,
\begin{enumerate}
\item  $\Am_l \Am_l^T$ and $\Am_l^T \Am_i$ are idempotent matrices.
\item $\Am_l \Am_l^T \Am_l = \Am_l$.
\end{enumerate}
\begin{proof} 
Follows from the properties of the sigma basis established in the \Cref{thm: prop_sig}. 
\end{proof}
\end{corollary}

If $\Am_l = \otimes_{k} \sigma_k,\sigma_k\in S$, then define $\bar{\Am}_l \coloneqq \otimes_k \bar{\sigma}_k$ where $\bar{\sigma}_k \in \{ \sigma_x, \Id \}$ and $\Am_l^c \coloneqq \bar{\Am}_l - \Am$.  It is trivial to see that $\bar{\Am}_l $ is a unitary matrix. \Cref{thm: Acomp} shows that $\bar{\Am}_l$ is nothing but the unitary completion of $\Am_l$, and $\Am_l^c$ is its unitary complement by using the \Cref{lem: AAc_orth}. 

\begin{lemma}
\label{lem: AAc_orth}
$\Am^T_l \Am^c_l = (\Am^c_l)^T \Am_l =  \Am_l (\Am^c_l)^T = \Am^c_l \Am_l^T = 0$.
\end{lemma}
\begin{proof}
First, note that
\begin{align*}
\bar{\Am}_l &= \otimes_k \bar{\sigma}_k = \otimes_k (\sigma_k + \sigma_k^c)=\Am_l+\sum_k \tilde{\Am}_k, 
\end{align*}
where, $\tilde{\Am}_k=\tilde{\sigma}_{k1} \otimes \dots \otimes \tilde{\sigma}_{kp} \otimes \dots \tilde{\sigma}_{kn}$, with at least one $\tilde{\sigma}_{kp}=\sigma^c_p$ corresponding to $\sigma_p$ term in $\Am_l$.
Thus, 
\begin{align*}
&\Am_l^T \Am_l^c = \sum_k (\sigma_1^T \otimes \dots \otimes \sigma_n^T) \tilde{\Am}_k \\
&= \sum_k (\sigma_1^T \tilde{\sigma}_{k1} \otimes \dots \otimes \sigma_j^T \tilde{\sigma}_{kj} \otimes \dots \otimes \sigma_n^T \tilde{\sigma}_{kn})=0,
\end{align*}
because, $\sigma_p^T \tilde{\sigma}_{kp} = \sigma_p^T \sigma^c_{p}=0$ for at least one $p\in\{1,\cdots,n\}, \forall k$.
The other equalities can be shown in a similar fashion. 
\end{proof}

\begin{theorem}
\label{thm: Acomp}
If $\Am_l = \otimes_{k} \sigma_k$ where $\sigma_k \in S$, then $\bar{\Am}_l \coloneqq\otimes_k \bar{\sigma}_k$ where $\bar{\sigma}_k = \sigma_x$ for $\bar{\sigma}_k \in \{ \Sp, \Sm\}$ and $\bar{\sigma}_k = \Id$ for $\bar{\sigma}_k \in \{ \Id, \Sp\Sm, \Sm\Sp\}$ is the unitary completion of $\Am_l$.
\end{theorem}
\begin{proof}
Consider a non-trivial subspace $W = \{ w: w = \otimes_k w_k \,\forall w_k \in W_k \} \subset \mathbb{R}^N$ where $W_k$ is the subspace associated with $\sigma_k$. See proof of  \Cref{thm: prop_sig}  for the construction of the subspaces $W_k$. Using Kronecker product rules, for $\forall \, w \in  W$, $ \bra{w} \Am_l^T \Am_l \ket{w} = \langle w \ket{w} $. We need to show that $\bar{\Am}_l \ket{w} = \Am \ket{w},$ $\forall w \in W$ and spans $\mathbb{R}^N$. Alternately, one can show that $\Am^c_l \ket{w} =0, \forall w \in W$, i.e., 
\begin{align*}
\|\Am_l^c \ket{w}\|_2^2 &= \bra{w} (\Am_l^c)^T \Am_l^c \ket{w} \\
&= \bra{w} \bar{\Am}_l^T \bar{\Am}_l \ket{w} - \bra{w} \Am_l^T \Am_l \ket{w} \\
 &= \bra{w} \Id \ket{w} - \langle w \ket{w} =0. 
\end{align*} 
\end{proof}

We are now ready to construct a quantum circuit for $\Um_l$. \Cref{thm: circuit_U} proves that $\Um_l$ can be constructed with at most \textit{n} single qubit gates and one $C^nX$ gate (\textit{n}-controlled Toffoli gate). It also provides a construction of the circuit for any $\Um_l$ of the form \Cref{eq: Um_l}. 

\begin{theorem}
\label{thm: circuit_U}
${\Um}_l$  as defined in (\ref{eq: Um_l}) can be implemented using at most $n$ single qubit gates and a single $C^n X$ gate
\end{theorem}
\begin{proof}
${\Um}_l$ can be written as a product of two unitary matrices ${\Um}_{l,1}, {\Um}_{l,2}$, i.e. ${\Um}_l={\Um}_{l,1}, {\Um}_{l,2}$,   such that 
\begin{align*}
{\Um}_{l,2} &= \Id \otimes \bar{\Am}_l = \begin{bmatrix}
\bar{\Am}_l &  \\
 & \bar{\Am}_l \\
 \end{bmatrix}  \\
 {\Um}_{l,1} &= {\Um}_l {\Um}_{l,2}^T =  \begin{bmatrix}
\Am^c_l & \Am_l \\
\Am_l & \Am^c_l
\end{bmatrix} \begin{bmatrix}
\bar{\Am}^T_l &  \\
 & \bar{\Am}^T_l \\
 \end{bmatrix}  \\
  &=
 \begin{bmatrix}
\Am^c_l (\Am^c)^T_l & \Am_l \Am^T_l \\
\Am_l \Am^T_l & \Am^c_l (\Am^c)^T_l \\
\end{bmatrix} \, (\because \text{by } \Cref{lem: AAc_orth})\\
 &=
 \begin{bmatrix}
\Id - \Am_l \Am^T_l & \Am_l \Am^T_l \\
\Am_l \Am^T_l & \Id - \Am_l \Am^T_l
\end{bmatrix}.  
\end{align*}
Note that $\Am_l \Am_l^T$ is a binary diagonal matrix with 1's and 0's at the diagonal, since
\begin{align*}
\Am_l \Am_l^T &= \otimes_p \sigma_p \sigma_p^T, \quad \sigma_p \sigma_p^T \in \{\Sp\Sm, \Sm\Sp, \Id \}.
\end{align*}
Therefore, ${\Um}_{i,1}$ is a permutation matrix of size $2^{n+1} \times 2^{n+1}$. Any permutation matrix can be implemented using only Toffoli gates~\cite{nielsen2010quantum}.  In particular, only a single $C^n X$ gate is required to implement ${\Um}_{l,1}$. 

Let the binary representation of row $r = \sum_{p=1}^{n} 2^{n-p} q(p)$, where $q(p) \in \{0, 1\}$. If row $r$ of  $\Am_l \Am_l^T$ is non-zero, then permute row $r$ of 
$\Id_{2^{n+1}}$ to row $r' \coloneqq 2^{n} +r$. Note that the binary representation of rows $r$ and $r'$ differ only by one bit. Thus, this permutation can be performed using a single $C^nX$ gate. The choice of the control operation (ope-control or closed-control) on the \textit{n}-qubits depends on the rows to be permuted as follows.
\begin{align*}
(\Am_l \Am_l^T)_{r, r} = 1 \quad \text{iff } (\sigma_p \sigma_p^T)_{q(p), q(p)} &= 1, \quad \forall p=1,\dots, n.
\end{align*}
If $(\sigma_p \sigma_p^T)_{0,0} = 1$ \Bigg(i.e., $\sigma_p \sigma_p^T = \begin{bmatrix}
1 & 0\\
0 & *
\end{bmatrix}$\Bigg) it corresponds to open-control on the $p$-th qubit and if $(\sigma_p \sigma_p^T)_{1,1} = 1$ \Bigg(i.e., $\sigma_p \sigma_p^T = \begin{bmatrix}
* & 0\\
0 & 1
\end{bmatrix}$\Bigg) it corresponds to closed-control on the $p$-th qubit. If both $(\sigma_p \sigma_p^T)_{0,0} = (\sigma_p \sigma_p^T)_{1,1} = 1$ then no control is necessary. In this fashion, we can design a $C^nX$ gate to implement ${\Um}_{l,1}$. 

Finally, by \Cref{thm: Acomp}, ${\Um}_{l,2}$ is only involves tensor products of $\sigma_x, \Id$, and thus can be implemented efficiently using the single qubit gate $\sigma_x$.
\end{proof}

\section{Discussion}
\label{sec: discussion}

\subsection{Relation to unitary dilation}

Our technique can be seen as a form of unitary dilation as it works by constructing a unitary matrix that operates on vectors in a higher dimensional subspace. The standard unitary dilation on $\Am_l$ when it is a contraction is defined as follows.
\begin{definition}
Let $\Am$ be contraction, i.e., $\|\Am\|_2\leq 1$, then the unitary dilation $\Um_{\Am}$ of $\Am$ is defined as:
\begin{equation}
\Tilde{\Um}_{\Am}=\left(
            \begin{array}{cc}
              \Am & \Dm_{\Am^{*}} \\
              \Dm_{\Am} & -\Am^{*} \\
            \end{array}
          \right), \label{eq:ud}
\end{equation}
where, $\Dm_{\Am}=\sqrt{\In-\Am^{*}\Am}$. 
\label{def:dilation}
\end{definition}
Note that $\Dm_{\Am}$ is well defined as $\In-\Am^{*}\Am\geq 0$ is a positive semi-definite matrix under the contraction assumption. As per the Sz.-Nagy dilation theorem, every contraction on a Hilbert space has a unitary dilation which is unique up to an unitary equivalence \cite{schaffer1955unitary}.

The application of $\tilde{\Um}_l \coloneqq \tilde{\Um}_{\Am_l}$ on an ancilla system of the form $\ket{0} \ket{\psi}$ gives,
\begin{equation*}
\tilde{\Um}_{l}|0\ra|\psi\ra=|0\ra \Am_l  |\psi\ra+|1\ra \Dm_{\Am_l} |\psi\ra,
\end{equation*}
similar to the expression in \Cref{eq: Umact}. We next discuss the relationship between the two dilation matrices $\tilde{\Um}_l$ as defined above and the unitary completion matrix we introduced in \Cref{eq: Um_l} and show that our matrix $\Um_l$ requires a shallow circuit to implement. 

For $\Am_l$ given as tensor product over sigma basis, $(\In-\Am^*_l\Am_l)^2=\In-\Am^T_l\Am_l$ using \Cref{cor: prop_A}. Thus, we can simplify $\tilde{\Um}_l$ as
\begin{equation*}
\tilde{\Um}_{l}=\begin{bmatrix}
	\Am_l &\In-\Am_l\Am_l^T	\\
              \In-\Am^{T}_l\Am_l  & -\Am^{T}_l  \\
            \end{bmatrix}.
\end{equation*}
Note that $\tilde{\Um}_{l}$ is unitary for the sigma basis even without negative sign in the (2,2) block position. To simplify analysis, we ignore the negative sign as it is unitarily equivalent (up to a controlled $\sigma_z$ gate). With this, $\tilde{\Um}_{l}$ is again a permutation matrix and can be expressed as a sequence of Toffoli gates. 

Consider the case when, $\Am_l = \otimes_k \sigma_k$ where $\sigma_k = \{\Sp\Sm, \Sm\Sp, \Id \} \subset S$. Then $\Am^T_l = \Am_l$,  $\Am_l\Am_l^T = \Am_l$ and $\Id - \Am_l = \bar{\Am}_l - \Am_l = \Am_l^c$. Thus, $\tilde{\Um}_l = \Um_l(\sigma_x \otimes \Id)$. Thus, we need again require one $C^n X$ gate to implement $\tilde{\Um}_l$.
 
The other extreme case is when $\Am_l = \otimes_k \sigma_k$ where $\sigma_k \in \{\Sp, \Sm\} \subset S$. Then, $\Am_l$ is a binary matrix with a single non-zero entry. Let's say row $r = \sum_{p=1}^{n} 2^{n-p} q(p)$ of $\Am_l$ has that non-zero entry. Here, $q(p) \in \{0, 1\}$. Consider row $r' = 2^{n} + \sum_{p=1}^{n} 2^{n-p} (1-q(p))$. We can construct  $\tilde{\Um}_l$ by permuting rows $r$ and $r'$ of matrix  $\sigma_x \otimes \Id = \begin{bmatrix}
& \Id \\
\Id &
\end{bmatrix}$. Note that the binary representation of $r$ and $r'$ differ in every bit. Each $C^nX$ gate only permutes rows that differ by a single bit. Thus, in order to permute rows $r$ and $r'$, we need $n$ $C^nX$ gates and \textit{another} set of $n-1$ $C^nX$ gates to un-permute the rows that were disturbed.  Thus, this requires a total of $2n-1$ $C^nX$ gates. This is in contrast to our unitary completion based technique, where a single $C^nX$ gate was sufficient to encode the permutation matrix $\Um_{l,1}$, see \Cref{thm: prop_sig}.

To illustrate this, consider a simple example with $\Am_l = \Sm$. In this case, with a single control bit, we have a $C^1X = CNOT$ gate.  
\begin{align*}
\tilde{\Um}_l &= \begin{bmatrix}
\Sm  & \Sp\Sm\\
\Sm\Sp & \Sp\\
\end{bmatrix} = \begin{bmatrix}
0&0&1&0\\
1&0&0&0\\
0&0&0&1\\
0&1&0&0\\
\end{bmatrix} = \begin{bmatrix}
1&0&0&0\\
0&0&1&0\\
0&1&0&0\\
0&0&0&1\\
\end{bmatrix}(\sigma_x \otimes \Id). 
\end{align*}
\begin{align*}
\Um_l &=  \begin{bmatrix}
I-\Sm\Sp & \Sm\Sp\\
\Sm\Sp  & I-\Sm\Sp
\end{bmatrix} (\Id \otimes \sigma_x) = \begin{bmatrix}
1&0&0&0\\
0&0&0&1\\
0&0&1&0\\
0&1&0&0
\end{bmatrix} (\Id \otimes \sigma_x).
\end{align*}

$\tilde{\Um}_l $ corresponds to the SWAP gate (which can be implemented using 3 CNOT gates) and a NOT gate. On the other hand, $\Um_l$ can be implemented using a NOT and CNOT gate. The corresponding circuits are shown in the \Cref{fig: U1}. 

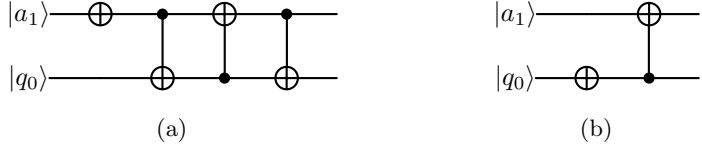
\begin{figure}
\centering
\begin{subfigure}[t]{0.35\textwidth}
\centering
\begin{quantikz}[wire types = {q, q}, transparent]
\ket{a_1} & \targ{} &\ctrl{1} & \targ{} &\ctrl{1} &  \\
\ket{q_0} &          & \targ{} & \ctrl{-1} & \targ{} & 
\end{quantikz} 
\caption{}
\end{subfigure}
~
\begin{subfigure}[t]{0.35\textwidth}
\centering
\begin{quantikz}[wire types = {q, q}, transparent]
\ket{a_1} & & \targ{} &  \\
\ket{q_0} & \targ{} &\ctrl{-1} &
\end{quantikz} 
\caption{}
\end{subfigure}
\caption{Circuits for implementing (a) $\tilde{\Um}_l$ and (b) $\bar{\Um}$  when $\Am_l = \Sm$.}
\label{fig: U1}
\end{figure}

In general, we expect our dilation operator based on unitary completion to be  more efficient than the standard unitary dilation operator. Future work will consider a more rigorous analysis to conclude this. 

\subsection{Relation to Pauli basis}

The Pauli matrices lie in the span of the sigma basis. Thus, the sigma basis set $S$ defined in \Cref{sec: method} is also an universal basis. Note that, out of the five basis matrices in $S$, only four are linearly independent, i.e. $S$ in an overcomplete basis. For example, $\Id = \Sp\Sm + \Sm\Sp$ and can be ignored. However, often times it is advantageous to keep $\Id$ in the basis set as it can lead to fewer terms in the decomposition as illustrated for the Heat equation matrix in~\Cref{subsec: heat}. Similarly, any Pauli matrix can be added to $S$ as relevant to the problem. The idea of constructing the unitary matrix in ~\Cref{eq: Um_l} can be extended for such cases as well. This gives one the flexibility to freely mix the Pauli and sigma basis terms to get an efficient tensor product decomposition for the problem at hand. 

Using Pauli basis, the number of terms $n_l$ in the LCU decomposition (\Cref{eq:LCU}) in general can vary from $N$ (diagonal matrices) to $N^2$ (dense matrix). However, for structured sparse matrix, e.g. arising in discretization of the Heat equation (see ~\Cref{eq: full_coeff}), there are no such theoretical estimates available. Therefore, we used the matrix splicing based numerical method proposed in \cite{lcunum} for the LCU decomposition of the matrix (\Cref{eq: full_coeff}) in the Pauli basis. ~\Cref{tab:lcu_terms} compares the number of terms for the Pauli and the number of terms (determined theoretically as discussed in the Section \Cref{subsec: heat}) for the sigma basis for different matrix size $N$. It is evident from these results that for structured sparse matrices one can expect the tensor product decomposition based on the sigma basis to provide significant advantage.

\begin{table}
\centering
\begin{tabular}{ccccc}\toprule
$n_x$ & $n_t$ & Matrix size ($N=n_xn_t$) & Pauli & Sigma \\
\midrule
4 & 4 & 16 & 26 & 19 \\
4 & 8 & 32 & 54 & 21 \\
8 & 8 & 64 & 102 & 25 \\
8 & 16 & 128 & 206 & 27 \\\bottomrule
\end{tabular}
\caption{Comparing the number of terms for the tensor product decomposition in Pauli basis and sigma basis for the Heat equation example.}
\label{tab:lcu_terms}
\end{table}

\subsection{Resource Estimation}
In our approach, the presence of an additional $C^nX$ gate in the quantum circuit for $\Um_l$ (equivalently, $\Am_l$), leads to  a deeper circuit for evaluating the local/global cost functions  as compared to using the Pauli basis. Unitary operators are often implemented using elementary universal gate sets such as CNOT and single-qubit gates. Efficient implementation of the $C^nX$ gate is an active area of research. Any implementation of the gate in terms of universal elementary gate sets must have a depth of $\Omega(\log n)$ and size (number of gates) $\Omega(n)$~\cite{10.5555/2011679.2011682}. In a recent paper~\cite{nie2024quantum}, the authors provide an implementation of the $C^nX$ gate using a quantum circuit consisting of single qubit gates and CNOT gates resulting in a circuit of depth atmost $\mathcal{O}(\log n)$ and size $\mathcal{O}(n)$ with 1 ancilla qubit. These results indicate that our quantum circuit can be implemented efficiently with only a minimal increase in the number of gates and circuit depth as compared to the Pauli basis approach. It is interesting to note that in future it may be possible to implement $C^nX$ gates directly on trapped ions, Rydberg atoms and superconducting circuit based quantum devices ~\cite{katz2022n}. 

\subsection{Comparison with Bell Measurement Approach}

The Bell measurement based quantum circuits (as discussed in the ~\Cref{sec: sigma}) require measurement of every qubit to compute the expectation values arising in global cost function. This increases the measurement overhead and cost of classical post-processing. On the other hand, our unitary completion approach relies, independent of the problem size, only on the measurement of two ancilla bits. Note that our approach requires an additional ancilla bit as compared to the Bell measurement approach. Furthermore, compared to our approach which can be applied for both the global/local cost function evaluation, there is no straightforward extension of the Bell measurement approach for computing the expectation terms $\delta_{ijk}$ arising in the local VQLS cost function. As discussed in the Introduction, local cost functions are advantageous to use in VQAs due to their resiliency to barren plateaus. 

\section{Conclusions}\label{sec: conc}
In this paper we proposed a novel approach for designing efficient quantum circuits for evaluating the global/local VQLS cost functions for matrices given in form of LCU type tensor product decomposition in the sigma basis. The sigma basis better exploits the sparsity and underlying structure of the matrices such as those arising from PDE discretization (e.g. Heat equaiton), leading to number of tensor product terms which scale only logarithmically with respect to the matrix size. This provides in worst case an exponential advantage over conventionally used Pauli basis. Given the sigma basis is comprised of non-unitary operators, we employed the concept of unitary completion to design quantum circuits for computing the VQLS global/local cost functions, and estimated the quantum resources needed. We compared our approach with other related concepts including unitary dilation and quantum Bell measurements, and discussed various pros/cons. In future it will worthwhile to explore other VQA applications where sigma type basis can be used to gain efficiency in computation of associated cost functions. 

\section{Acknowledgments}
This research was developed with funding from the Defense Advanced Research Projects Agency (DARPA). The views, opinions, and/or findings expressed are those of the author(s) and should not be interpreted as representing the official views or policies of the Department of Defense or the U.S. Government.

\bibliographystyle{unsrt}
\bibliography{references}

\end{document}